\documentclass[12pt]{article}
\usepackage{mathrsfs}
\usepackage{dsfont}
\usepackage{amsthm}
\usepackage{mathrsfs}
\usepackage{amsmath}
\usepackage{amsfonts}
\usepackage[colorlinks,linkcolor=black,anchorcolor=blue,citecolor=blue]{hyperref}
\usepackage{amssymb, amsmath, cite}

\newtheorem{theorem}{Theorem}[section]
\newtheorem{lemma}{Lemma}[section]

\newcommand\be{\begin{equation}}
\newcommand\ee{\end{equation}}
\newcommand\ber{\begin{eqnarray}}
\newcommand\eer{\end{eqnarray}}
\newcommand\berr{\begin{eqnarray*}}
\newcommand\eerr{\end{eqnarray*}}
\newcommand\bea{\begin{eqnarray}}
\newcommand\eea{\end{eqnarray}}

\newcommand{\nn}{\nonumber}

\newcommand{\dd}{\mathrm{d}}
\newcommand\e{\mathrm{e}}\newcommand\pa{\partial}
{\setlength\arraycolsep{2pt}

\begin{document}

\title{Existence of Julia--Zee dyon and 't Hooft--Polyakov monopole with new field strength tensor}
\author{Lei Cao, Yilu Xu\footnote{E-mail address:
104753180630@vip.henu.edu.cn.}\\School of Mathematics and Statistics, Henan University\\
Kaifeng, Henan 475004, PR China}
\date{}
\maketitle

\begin{abstract}

In this paper, we use a modified Abelian field strength tensor in Georgi--Glashow model and obtain a new Julia--Zee dyon equations which degenerated into the 't Hooft--Polyakov monopole equations when the profile function $J=0$. Combining a three--step iterative shooting argument and a fixed--point theorem approach, we establish the existence of the static solution of the Julia--Zee dyon equations and discuss its qualitative properties. In addition, we show that the total magnetic charge is a constant related to the gauge coupling constant and the total electric charge depend continuously on the profile functions.
\end{abstract}

\begin{enumerate}

\item[]

\item[]
{MSC numbers(2020):} 34B15, 34B40, 81T13

\end{enumerate}

\section{Introduction}\label{s1}
\setcounter{equation}{0}

Fifty years ago, 't Hooft \cite{Ho} and Polyakov \cite{Po} proposed a model for a magnetic monopole arising as a static solution of the classical equations for the $SU(2)$ Yang--Mills field coupled to an $SU(2)$ Higgs field. The model was later extended by Julia and Zee \cite{Ju} so that the monopole becomes a dyon, possessing both electric and magnetic charge. Besides, they also show that the finite--energy static solutions in the classical relativistic gauge field theory over the $(2+1)$--dimensional Minkowski spacetime must be electrically neutral. This statement is well known as Julia--Zee theorem which leads to many interesting consequences. For example, it makes it transparent that the static Abelian Higgs model is exactly the Ginzburg--Landau theory \cite{Gin} which is purely magnetic \cite{Ja,Nie}.

The solution of the 't Hooft--Polyakov--Julia--Zee model has long been studied with increasing enthusiasm since the work of Julia and Zee. When the coupling constant $\lambda=0$, Prasad and Sommerfield \cite{Pra} presented an exact Julia--Zee dyon solution to the nonlinear field equations which describe a classical excitation possessing magnetic and electric charge. This solution has finite energy and exhibits explicitly those properties which have previously been found by numerical analysis. Kasuya and Kamata \cite{Ka} studied the spherically symmetric Julia--Zee dyon solution of the coupled Yang--Mills--Higgs systems in curved space--time, and gave an exact solution for which the space--time metric takes the Reissner--Nordstr\"{o}m form. McLeod and Wang \cite{McL} showed the existence of solutions for the 't Hooft--Polyakov--Julia--Zee monopole when the coupling constant $\lambda=1$.

To throw more light on the mechanical properties of the monopole and the dyon and on the local stability requirement, Panteleeva \cite{Pan} eliminates the long--range contribution from the energy--momentum tensor of the monopole and the dyon. While such calculation resulting contribution cannot be uniquely defined. A modified 't Hooft definition of the Abelian field strength tensor is introduced to overcome this difficulty. Using the modified Abelian field strength tensor in Georgi--Glashow model, we obtain the new equations for the profile functions, whose solution contains both electric and magnetic so is a dyon. In particular, taking the profile function $J=0$, we get equations of motion describing only magnetic. In this paper, we analyse these equations to show that the corresponding boundary value problems possess a dyon or monopole solution and discuss its qualitative properties.

The paper is organised as follows. In Sec. 2, we present the Julia--Zee dyon equations and 't Hooft--Polyakov monopole equations of J. Yu. Panteleeva and state our main results for the corresponding solutions. Sec. 3 contains our results concerning the proof of the existence and uniqueness of the solution of each governing equation and the qualitative properties. In particular, we use a fixed point theorem to obtain a general solution to the corresponding boundary value problem. We also calculate the total magnetic charge and electric charge.

\section{Julia--Zee dyon and 't Hooft--Polyakov monopole}\label{s2}
\setcounter{equation}{0}

The simplest model in which magnetic monopoles exist is the gauge $SU(2)$ Georgi--Glashow model with Higgs triplet field $\phi^a, a=1,2,3$, which belongs to adjoint representation. The corresponding gauge invariant action of the model is
\be\label{2.1}
S=\int\dd^4 x\left[-\frac 14 F_{\mu\nu}^aF^{\mu\nu a}+\frac 12(D_\mu\phi)^a(D^\mu\phi)^a-\frac{\lambda}{4}(\phi^a\phi^a-v^2)^2\right],
\ee
where $\lambda$ is a dimensionless coupling constant, $v^2$ is the squared vacuum expectation value of the Higgs field, $\mu$ and $\nu$ are Lorenz indices. The covariant derivative and the non--Ablian field strength tensor are defined as follows
\bea
(D_\mu\phi)^a&=&\pa_\mu\phi^a+g\epsilon^{abc}A_\mu^b\phi^c,\label{2.1a}\\
F_{\mu\nu}^a&=&\pa_\mu A_\nu^a-\pa_\nu A_\mu^a+g\epsilon^{abc}A_\mu^bA_\nu^c,\label{2.1b}
\eea
where $g$ is the gauge coupling constant and $A_\mu^a$ is the gauge vector field. We introduce the $U(1)$ field strength tensor $\mathcal{F}_{\mu\nu}$ as
\be\label{2.1c}
\mathcal{F}_{\mu\nu}=\frac{\phi^a}{v}F_{\mu\nu}^a-\frac{1}{gv^3}\varepsilon^{abc}\phi^aD_\mu\phi^bD_\nu\phi^c,
\ee
which is smooth as the Faddeev's \cite{Fa} and Boulware's \cite{Bo} definitions and coincides with the 't Hooft's \cite{Ho} definition at long distances.

Set $r=|x|$. In the static soliton configuration with finite energy, choosing the spherically--symmetric ansatz
\bea
\phi^a&=&\frac{x^a}{r}v h(r),\label{2.1d}\\
A_i^a&=&\frac{1}{g r^2}\epsilon_{aib}x^b(1-F(r)),\label{2.1e}\\
A_0^a&=&\frac{J(r)x^a}{gr^2}.\label{2.1f}
\eea
Then the equations of motion of \eqref{2.1} can be derived as follows
\bea
&&D_0F_{00}^a-D_i F_{i0}^a+g\epsilon^{abc}\phi^bD_0\phi^c=0,\label{2.1g}\\
&&D_0F_{0k}^a-D_i F_{ik}^a-g\epsilon^{abc}\phi^cD_k\phi^b=0,\label{2.1h}\\
&&D_0D_0\phi^a-D_i D_i\phi^a+\lambda(\phi^b\phi^b-v^2)\phi^a=0.\label{2.1i}
\eea
In terms of the profile functions with the dimensionless variable $\rho=gvr$, and using prime to denote differentiation with respect to the radial variable $\rho$, we have
\bea
&&F''\rho^2=(F^2-1)F-(J^2-\rho^2h^2)F,\label{2.2}\\
&&h''\rho^2+2h'\rho=2F^2h+\frac{\beta^2}{2}h(h^2-1)\rho^2,\label{2.3}\\
&&J''\rho^2=2JF^2,\label{2.4}
\eea
where $\beta^2=2\lambda/g^2$. With respect to the regularity requirement and the finite energy, $F, h, J$ must satisfy
\bea
F(\rho)&=&1,~~~h(\rho)=0,~~~J(\rho)=0,~~~\text{as}~~\rho\to 0,\label{2.5}\\
F(\rho)&=&0,~~~h(\rho)=1,~~~J(\rho)/\rho=C,~~~\text{as}~~\rho\to \infty.\label{2.6}
\eea
We assume as \cite{Pan} that $0\leq C<1$.

Inserting \eqref{2.1d}--\eqref{2.1f} into \eqref{2.1c} and using \eqref{2.1a} and \eqref{2.1b}, we find the electric and magnetic fields, $\mathbf{E}=(E^i)$ and $\mathbf{B}=(B^i)$, as follows,
\bea
E^i&=&-\mathcal{F}^{0i}=\frac{x^i h(r)}{gr}\frac{\dd}{\dd r}\left(\frac{J(r)}{r}\right),\label{2.13}\\
B^i&=&-\frac 1 2\epsilon_{ijk}\mathcal{F}^{jk}=\frac{x^i}{2g|r|^3\left(h(r)(1-F^2(r))+h^3(r)F^2(r)\right)}.\label{2.14}
\eea

From \eqref{2.13} it can be seen that, if $J=0$, then $\mathbf{E}=\mathbf{0}$ and there is no electric field. By integrating \eqref{2.14}, the magnetic charge $q_m$ can be obtained
\bea\label{2.15}
q_m&=&\frac{1}{4\pi}\lim_{r\to\infty}\oint_{|x|=r}\mathbf{B}\cdot \dd \mathbf{S}\nn\\
&=&\frac{1}{2g}\lim_{r\to\infty}\frac{1}{h(r)(1-F^2(r))+h^3(r)F^2(r)}=\frac{1}{2g}.
\eea

If we fix zero component of a vector field through the gauge $A_0^a=0$ to have zero electric field, combining the spherically--symmetric ansatz \eqref{2.1d}--\eqref{2.1e}. Then we can get the equations of motion that can be computed by varying the action with respect to scalar and vector fields, which for the static case reduce to the following form
\bea
&&D_iF_{ij}^a=g\epsilon^{abc}\phi^b(D_j\phi)^c,\label{2.1j}\\
&&D_i(D_i\phi)^a=\lambda(\phi^b\phi^b-v^2)\phi^a.\label{2.1k}
\eea
Furthermore, the equations of the profile functions \eqref{2.2}--\eqref{2.4} are simplified to the following form
\bea
&&F''\rho^2=(F^2-1)F+\rho^2h^2F,\label{2.16}\\
&&h''\rho^2+2h'\rho=2F^2h+\frac{\beta^2}{2}h(h^2-1)\rho^2,\label{2.17}
\eea
and the corresponding boundary conditions read
\bea
F(\rho)&=&1,~~~h(\rho)=0,~~~\text{as}~~\rho\to 0,\label{2.18}\\
F(\rho)&=&0,~~~h(\rho)=1,~~~\text{as}~~\rho\to \infty.\label{2.19}
\eea

For the nonlinear ordinary differential equations \eqref{2.2}--\eqref{2.4} and \eqref{2.16}--\eqref{2.17} with the corresponding boundary conditions, inspired by the literature \cite{McLeod1,Bizon,Hastings,Smoller,Wang,Mcleod2,Chen}, we develop the methods and techniques in which we establish the existence of the solutions and study the related properties of the solutions.

Then the main results of this paper are stated as follows.

\begin{theorem}\label{th2.1}
For any $\lambda>0$ and $0\leq C<1$, the non--Abelian gauge field equations \eqref{2.1g}--\eqref{2.1i} have a finite--energy static solution $(\mathbf{A_{\mu}},\phi)$ defined by \eqref{2.1d}--\eqref{2.1f} so that the obtained solution configuration functions of the radially symmetric Julia-Zee dyon equations \eqref{2.2}--\eqref{2.4} have the following properties

{\rm(i)} The function $F(\rho),h(\rho),J(\rho)\in C^{1}\left([0,+\infty)\right)$ and $F'(0)=J'(0)=0${\rm;}

{\rm(ii)} $0\leq F(\rho),h(\rho)\leq1,$ $0\leq J'(\rho)\leq C,\,\forall\rho\geq 0{\rm;}$ $h(\rho)\,,J(\rho)$ increasing, $F(\rho)$ decreasing{\rm;}

{\rm(iii)} $\rho^{-2}(1-F(\rho))\,,\rho^{-2}J(\rho)\,,\rho^{-1}h(\rho)$ decreasing and all bounded as $\rho\rightarrow 0{\rm;}$

{\rm(iv)} $\lim\limits_{r\to0}F(\rho)=1,~\lim\limits_{r\to0}h(\rho)=0,~\lim\limits_{r\to0}J(\rho)=0{\rm;}$

{\rm(v)} $F(\rho)=O(\e^{-(1-\varepsilon)\rho}), h(\rho)=1+O(\e^{-\beta(1-\varepsilon)\rho}), J''(\rho)=O(\e^{-2(\sqrt{1-C^2}-\varepsilon)\rho})$
as $\rho\to\infty$, where $0<\varepsilon<1$ is arbitrary{\rm;}

{\rm(vi)} This solution carries a electric charge $q_e$,
\bea
q_e=\frac{1}{4\pi}\lim_{r\to\infty}\oint_{|x|=r}\mathbf{E}\cdot \dd \mathbf{S}=\frac{1}{g}\int_0^\infty \left(\frac{2F^2Jh}{r}-h'J+rh'J'\right)\dd r,
\eea
and a magnetic charge $q_m=1/2g$.
\end{theorem}

\begin{theorem}\label{th2.2}
The 't Hooft-Polyakov monopole equations \eqref{2.1j}--\eqref{2.1k} have a family of static finite energy smooth solutions which satisfy the radial symmetry properties given in \eqref{2.1d}--\eqref{2.1e} and $A_0^a=0$. The obtained solution configuration functions $(F(\rho), h(\rho))$ have the properties that $F=F(\rho)$ and $h=h(\rho)$ are monotone functions whose ranges are strictly confined between the corresponding limiting values of $F(\rho)$ and $h(\rho)$ at $\rho=0$ and $\rho=\infty$, respectively, as stated in \eqref{2.18}--\eqref{2.19}.
\end{theorem}

Since the proof of Theorem \ref{th2.2} is similar to that of Theorem \ref{th2.1}, we will omit the details.

As the end of this section, in order to prove the theorem \ref{th2.1}, we state the arrangement of next section as follows. In Section 3,
a three--step iterative approach is employed to show the existence of the Julia--Zee dyon solutions. We first give a series of lemmas as the primary works to show the existence and uniqueness of the Julia--Zee dyon solutions of each nonlinear ordinary differential equation. In addition, properties of the solutions are established. Finally, by applying the Schauder fixed--point theorem, we solve the two--point boundary value problem \eqref{2.2}--\eqref{2.6}.

\section{Proof of existence result}\label{s3}
\setcounter{equation}{0}

In this section, we start our construction of a solution pair of \eqref{2.2}--\eqref{2.4} subject to the boundary conditions \eqref{2.5}--\eqref{2.6} by a multiple shooting method.

\subsection{The equation governing the $F$--component}

In this subsection, we consider the $F$--component equation \eqref{2.2} for a pair of suitable fixed functions $J, h$. This construction will be performed in four steps.

\begin{lemma}\label{le3.1}
Given a pair of functions $J(\rho),\,h(\rho)$ in $C\left([0,+\infty)\right)$ such that $J(\rho),\,h(\rho)$ increasing, $0<J(\rho),\,\rho h(\rho)\leq \rho^{1+k}R^{*}$ for all $0<\rho\leq1$ and $J(0)=h(0)=0,\,h(\infty)=1,$ $J(\rho)/\rho=C$ as $\rho\to \infty$, where $0\leq C<1,$ $0<k<1$ and $R^{*}$ is a positive constant. We can find a unique continuously differentiable solution $F(\rho)$ satisfying $\eqref{2.2}$ subject to the conditions $F(0)=1,\,F(\infty)=0$, along with $F'(0)=0$, $\rho^{-2}(F-1)(\rho)$ increasing in $\rho$ and bounded as $\rho\to 0$. Besides, $\rho^{-2}(F-1)(\rho)\leq R^{**}$ for $\rho<1$, where $R^{**}=R^{**}(R^{*},k,N)>0$ and $N$ is a positive constant. Furthermore, we have the sharp asymptotic estimates
\bea
F(\rho)&=&1+a\rho^2+O(\rho^{2k+2}),~~-{(N{R^{*}}^{2})^{\frac{1}{1+k}}}<a<0,~~0<k<1,~~\rho\to 0,\label{00}\\
F(\rho)&=&O(\e^{-(1-\varepsilon)\rho}),~~~\rho\to\infty,~~~\varepsilon>0.\label{11}
\eea
\end{lemma}

\begin{proof}
In order to show the existence and uniqueness of the equation $\eqref{2.2}$ with boundary condition $F(0)=1,\,F(\infty)=0$, We split the proof into four steps.

Step 1{\rm:} The existence and uniqueness of the local solution to the initial value problem. If we set $f(\rho)=F(\rho)-1$, then since $f(0)=0$, equation $\eqref{2.2}$ can be written in the form
\begin{eqnarray}
\label{3.1}
f''=\frac{2f}{\rho^{2}}+h^{2}(f+1)+\frac{f^{3}+3f^{2}-J^{2}(f+1)}{\rho^{2}},\,\,\,\,\rho>0.
\end{eqnarray}
It is worth noting that the two linearly independent solutions of equation
\begin{eqnarray}
\label{3.2}
f''-\frac{2f}{\rho^{2}}=0,\,\,\,\,\rho>0
\end{eqnarray}
are $\rho^{2}$ and $\rho^{-1}$. Since $f(0)=0$, we claim that the solution of equation $\eqref{3.1}$ must satisfy $f(\rho)=O\left(\rho^{2}\right)\,(\rho\rightarrow0)$. With $(J(\rho),h(\rho))$ given, we stress that $f(\rho)$ may locally be labeled by an initial parameter $a$ as follows
\begin{eqnarray}
\label{3.3}
f''(\rho;a)=\frac{2f(\rho;a)}{\rho^{2}}+h^{2}(f(\rho;a)+1)+\frac{f^{3}(\rho;a)+3f^{2}(\rho;a)-J^{2}(f(\rho;a)+1)}{\rho^{2}},
\end{eqnarray}
\begin{eqnarray}
\label{3.4}
f(\rho;a)=a\rho^{2}+O\left(\rho^{2k+2}\right),\,\,0<k<1,\,\,\rho\rightarrow0,
\end{eqnarray}
where $f(\rho;a)$ could be used to denote the dependence of the local solution to $\eqref{3.1}$ on a certain constant $a<0$. Such parameter $a$ may be viewed as a shooting parameter in our method. Next, we prove the judgment that $\eqref{3.1}$ has a unique local solution near $\rho=0$ with $\eqref{3.4}$ for any fixed $a<0$ and $(J(\rho),h(\rho))$ given.

Using the basic theory of ordinary differential equations, two linearly independent solutions $\rho^{2}$ and $\rho^{-1}$ of $\eqref{3.2}$ and the condition of $f(0;a)=0$, we can transform $\eqref{3.3}$ into the integral form
\begin{eqnarray}
\label{3.5}
f(\rho;a)\!=\!a\rho^{2}\!\!+\!\!\frac{1}{3}\int_0^\rho\bigg(\frac{\rho^{2}}{r}\!-\!\frac{r^{2}}{\rho}\bigg)\bigg\{\bigg(h^{2}(r)\!-\!\frac{J^{2}(r)}{r^{2}}\bigg)(f(r;a)\!+\!1)\!+\!\frac{3f^{2}(r;a)\!\!+\!\!f^{3}(r;a)}{r^{2}}\bigg\}\dd r,
\end{eqnarray}
where $a$ is an arbitrary constant. The equation $\eqref{3.3}$ can now be easily solved according to Picard iteration as follows, at least near $\rho=0$. Set $f_{0}(\rho)=a\rho^{2}$ and
\begin{eqnarray}
\label{3.6}
g(\rho,r,f_{0}(r))=\frac{1}{3}\bigg(\frac{\rho^{2}}{r}-\frac{r^{2}}{\rho}\bigg)\bigg\{\bigg(h^{2}(r)-\frac{J^{2}(r)}{r^{2}}\bigg)({f_{0}}(r)+1)+\frac{3{f_{0}}^{2}(r)+{f_{0}}^{2}(r)}{r^{2}}\bigg\},
\end{eqnarray}
where $\rho,r>0.$
Assume
\begin{equation*}
  f_{n+1}(\rho)=a\rho^{2}+\int_0^\rho g(\rho,r,f_{n}(r))dr,~~~n=0,1,2,\cdots
\end{equation*}
Indeed, since the assumption on $J(\rho)$ and $h(\rho)$, we have
\begin{eqnarray}
\label{3.7}
|f_{1}(\rho)-f_{0}(\rho)|&\leq&\int_0^{\rho}\left|g(\rho,r,f_{0}(r))\right|\dd r\notag\\[1mm]
&\leq&\frac{2}{3}\int_0^{\rho}\frac{\rho^{2}}{r}\left(2r^{2k}{R^{*}}^{2}+3|a|^{2}r^{2}\right)\dd s\notag\\[1mm]
&\leq&\frac{2}{3k}{R^{*}}^{2}\rho^{2k+2}+|a|^{2}\rho^{4}
\leq M_{1}\rho^{{2k+2}},~~~\rho\leq1,
\end{eqnarray}
where $M_{1}=\frac{2}{3k}{R^{*}}^{2}+|a|^{2}$. Therefore,
\begin{eqnarray}
\label{3.8}
|f_{1}(\rho)|\leq M_{1}\rho^{2k+2}+|a|\rho^{2} \leq M_{2}\rho^{2},~~~\rho\leq\min\bigg\{1,\sqrt{\frac{1}{|a|}}\bigg\},
\end{eqnarray}
where $M_{2}=\frac{2}{3k}{R^{*}}^{2}+|a|^{2}+|a|$. Then, for $\rho\leq\min\left\{1,\sqrt{\frac{1}{4M_{2}}},\sqrt{\frac{1}{|a|}}\right\}$, we get
\begin{eqnarray}
\label{3.9}
|f_{2}(\rho)-f_{0}(\rho)|&\leq&\int_0^{\rho}\left|g(\rho,r,f_{1}(r))\right|\dd r\notag\\[1mm]
&\leq&\frac{2}{3}\int_0^{\rho}\frac{\rho^{2}}{r}\left|2r^{2k}{R^{*}}^{2}+\frac{3f_{0}^{2}(r)}{r^{2}}+\frac{3(f_{1}^{2}(r)-f_{0}^{2}(r))}{r^{2}}\right|\dd r\notag\\[1mm]
&\leq&\frac{2}{3}\int_0^{\rho}\frac{\rho^{2}}{r}\left(2r^{2k}{R^{*}}^{2}+3|a|^{2}r^{2}\right)dr+4M_{2}\int_0^{\rho}\frac{\rho^{2}}{r}|f_{1}(r)-f_{0}(r)|\dd r\notag\\[1mm]
&\leq&M_{1}\rho^{2k+2}+4M_{2}M_{1}\rho^{2k+4}\leq2M_{1}\rho^{2k+2}
\end{eqnarray}
and
\begin{eqnarray}
\label{3.10}
|f_{2}(\rho)|\leq M_{2}\rho^{2}.
\end{eqnarray}
Generally, for $\rho\leq\min\left\{1,\sqrt{\frac{1}{4M_{2}}},\sqrt{\frac{1}{|a|}}\right\}$, we obtain
\begin{eqnarray}
\label{3.11}
|f_{n+1}(\rho)-f_{0}(\rho)|&\leq&\int_0^{\rho}\left|g(\rho,r,f_{n}(r))\right|\dd r\notag\\[1mm]
&\leq&\frac{2}{3}\int_0^{\rho}\frac{\rho^{2}}{r}\left|2r^{2k}{R^{*}}^{2}+\frac{3f_{0}^{2}(r)}{r^{2}}+\frac{3(f_{n}^{2}(r)-f_{0}^{2}(r))}{r^{2}}\right|\dd r\notag\\[1mm]
&\leq&\frac{2}{3}\int_0^{\rho}\frac{\rho^{2}}{r}\left(2r^{2k}{R^{*}}^{2}+3|a|^{2}r^{2}\right)dr+4M_{2}\int_0^{\rho}\frac{\rho^{2}}{r}|f_{n}(r)-f_{0}(r)|\dd r\notag\\[1mm]
&\leq&M_{1}\rho^{2k+2}+4M_{2}M_{1}\rho^{2k+4}\leq2M_{1}\rho^{2k+2},~~~n=1,2,\cdots
\end{eqnarray}
and
\begin{eqnarray}
\label{3.12}
|f_{n}(\rho)|\leq M_{2}\rho^{2}.
\end{eqnarray}
Hence,
\begin{eqnarray}
\label{3.13}
&&|f_{n+1}(\rho)-f_{n}(\rho)|\notag\\[1mm]
&&\leq\int_0^{\rho}\left|g(\rho,r,f_{n}(r))-g(\rho,r,f_{n-1}(r))\right|\dd r\notag\\[1mm]
&&=\frac{1}{3}\int_0^{\rho}\left(\frac{\rho^{2}}{r}-\frac{r^{2}}{\rho}\right)\left|2r^{2k}{R^{*}}^{2}+\frac{f_{n}^{2}+f_{n-1}^{2}+f_{n}f_{n-1}+3(f_{n}+f_{n-1})}{r^{2}}\right|\left|f_{n}-f_{n-1}\right|\dd r\notag\\[1mm]
&&\leq\frac{2}{3}\int_0^{\rho}\frac{\rho^{2}}{r}\left(2r^{2k}{R^{*}}^{2}+3M_{2}^{2}r^{2}+6M_{2}\right)\left|f_{n}-f_{n-1}\right|\dd r\notag\\[1mm]
&&\leq M_{3}\int_0^{\rho}\frac{\rho^{2}}{r}\left|f_{n}-f_{n-1}\right|\dd r,
~~~\rho\leq\min\bigg\{1,\sqrt{\frac{1}{4M_{2}}},\sqrt{\frac{1}{|a|}}\bigg\},
~~~n=1,2,\cdots
\end{eqnarray}
where $M_{3}=\frac{4}{3}{R^{*}}^{2}+2M_{2}^{2}+4M_{2}$. Set
\begin{eqnarray}
\label{3.14}
\delta=\min\bigg\{1,\sqrt{\frac{1}{4M_{2}}},\sqrt{\frac{1}{2M_{3}}},
\sqrt{\frac{1}{|a|}}\bigg\}
\end{eqnarray}
and iterating $\eqref{3.13}$, we have
\begin{eqnarray}
\label{3.15}
|f_{n+1}(\rho)-f_{n}(\rho)|\leq\prod_{i=1}^{n}\frac{1}{2(k+i)}
M_{1}\rho^{2k+2}(M_{3}\rho^{2})^{n},~~~\rho\in[0,\delta],~~~n=1,2,\cdots.
\end{eqnarray}
Thus,$\{f_{n}(\rho)\}$ uniformly converges in $[0,\delta]$. Hence, the limit $f(\rho)=\lim_{n\rightarrow\infty}f_{n}(\rho)$ is a solution of $\eqref{3.5}$ over $[0,\delta]$. The local uniqueness of the solution to $\eqref{3.5}$ is obvious. Finally, letting $n\rightarrow\infty$ in $\eqref{3.11}$, then \eqref{3.4} is obtained. Thus the existence and uniqueness of \eqref{3.1} near origin has been proved.

A standard theorem on continuation of solutions assures us that the solution
can be continued as a function of $\rho$ until $f$ becomes infinite. Further, applying the continuous dependence of the solution on the parameters theorem we acquire that the solution $f(\rho;a)$ depends continuously on the parameter $a$.

Step 2{\rm:} The existence of the global solution to the boundary value problem. According to the standard theorem on continuation of solutions, the solution $f(\rho;a)$ as a function of $\rho$ can be extended to its maximal existence interval $[0,R_{a})$ with $R_a$ being infinite or finite.

With the discussion of Step 1, we are interested in $a<0$. Thus, we define $\mathcal{A}=\{a|a<0\}$ and the shooting sets $\mathcal{A}_{1},\mathcal{A}_{2},\mathcal{A}_{3}$ as follows
\begin{eqnarray*}
\mathcal{A}_{1}&=&\left\{a<0:\,f'(\rho;a)\,\,\mbox{becomes positive before}\,\,f(\rho;a)\,\,\mbox{reaches}\,-1\right\},\\[1mm] \mathcal{A}_{2}&=&\left\{a<0:\,f(\rho;a)\,\,\mbox{crosses}\,-1\,\,\mbox{before}\,\,f'(\rho;a)\,\,\mbox{becomes}\,\,0\right\},\\[1mm]
\mathcal{A}_{3}&=&\left\{a<0:\,\forall\rho>0,f'(\rho;a)\leq0,-1<f(\rho;a)<0\right\},
\end{eqnarray*}
where $f^{\prime}(\rho;a)=\frac{\partial f(\rho;a)}{\partial \rho}$ with $f(\rho;a)$ being the unique solution produced in Step 1.

From the construction of sets, we observe now that
\begin{equation*}
\mathcal{A}_{1}\cup\mathcal{A}_{2}\cup\mathcal{A}_{3}=\mathcal{A},\,\,\,
\mathcal{A}_{1}\cap\mathcal{A}_{2}=\mathcal{A}_{2}\cap\mathcal{A}_{3}=\mathcal{A}_{3}\cap\mathcal{A}_{1}=\emptyset.
\end{equation*}
In fact, it is obvious that $\mathcal{A}_{1}$, $\mathcal{A}_{2}$ and $\mathcal{A}_{3}$ do not intersect. Therefore, we are now ready to prove $\mathcal{A}_{1}\cup\mathcal{A}_{2}\cup\mathcal{A}_{3}=\mathcal{A}$. If $a<0$ but $a\notin\mathcal{A}_{1}$, then we consider three cases as follows. The first case is $\forall\rho>0,f'(\rho;a)\leq0,-1<f(\rho;a)<0$, which is equivalent to $a\in\mathcal{A}_{3}$. The second case is $\forall\rho>0,f'(\rho;a)\leq0$ and $f(\rho_{0};a)\leq-1$ for some $\rho_{0}$. Furthermore, it is worth stressing that $f(\rho_{0};a)=-1$ and $f'(\rho_{0};a)=0$ can not hold simultaneously. Suppose by contradiction, according to $\eqref{3.3}$, we have $f''(\rho_0;a)=0$. Then we get $f(\rho_{0};a)\equiv-1$, which contradicts $\eqref{3.4}$. Furthermore, we obtain $a\in\mathcal{A}_{2}$. The third case is that there is a first value $\rho_0$ such that $f'(\rho_0;a)=0$ and $f(\rho_{0};a)<-1$, which means $a\in\mathcal{A}_{2}$. In summary, If $a<0$ but $a\notin\mathcal{A}_{1}$, then $a\in\mathcal{A}_{2}\cup\mathcal{A}_{3}$. Consequently, we obtain the desired assertion.

In order to show that the set $\mathcal{A}_{3}$ is nonempty. We will next prove that the sets $\mathcal{A}_{1},\mathcal{A}_{2}$ are both open and nonempty.

First, $\mathcal{A}_{1}$ contains small $|a|$. Note that $\eqref{3.4}$ is true for all $a\in(-\infty,0]$ from Step 1, and $-1<f(\rho;a)<1,\,\forall \rho\in[0,\delta]$, where the definition of $\delta$ is in $\eqref{3.14}$. It is obvious that $\delta$ is positive as $a\rightarrow0$. Then inserting $|a|=0$ into the equation $\eqref{3.5}$, we have $f(\rho;a)>0,\,f'(\rho;a)>0,\,\forall\rho\in(0,\delta)$. Hence, for given $\rho_1\in(0,\delta)$, in view of the continuous dependence of solution on the parameter $a$, there exists a small $\varepsilon>0$ such that $f(\rho_1;a)>0,\,f'(\rho_1;a)>0,\,\forall a\in(-\varepsilon,0)$. In addition, for any given $a<0$, according to $\eqref{3.5}$, we obtain $f(\rho;a)<0,\,f'(\rho;a)<0$ initially. Then, there is a $\rho_2\in(0,\rho_1)\subset(0,\delta)$ satisfying $-1<f(\rho_2;a)<0,\,f'(\rho_2;a)=0$. Hence, there is a $\rho_3\in(\rho_2,\delta)\subset(0,\delta)$ so that $f'(\rho_3;a)>0$ and $-1<f(\rho;a)<1$ for $\rho\in(0,\rho_3]$. Consequently, $(-\varepsilon,0)\subset\mathcal{A}_{1}$ is nonempty. Applying the continuous dependence of the solution on the parameter $a$, we know that $\mathcal{A}_{1}$ is open.

Second, $\mathcal{A}_{2}$ contains large $|a|$. Here we introduce a transformation $t=|a|^{\frac{1}{2}}\rho$, then $\eqref{3.5}$ becomes
\begin{eqnarray}
\label{3.16}
f(t;a)\!\!=\!\!-t^{2}\!+\!\frac{1}{3}\int_0^t(\frac{t^{2}}{\tau}\!\!-\!\!\frac{\tau^{2}}{t})\bigg[\bigg(\frac{h^{2}(\tau)}{|a|}\!\!-\!\!\frac{J^{2}(\tau)}{\tau^{2}}\bigg)(f(\tau;a)\!+\!1)\!+\!\frac{3f^{2}(\tau;a)+f^{3}(\tau;a)}{\tau^{2}}\bigg]\dd\tau.
\end{eqnarray}
The corresponding differential form is
\begin{eqnarray}
\label{3.17}
f''(t;a)\!=\!\frac{2f(t;a)}{t^{2}}\!+\!\bigg[\frac{h^{2}\big(t|a|^{-\frac{1}{2}}\big)}{|a|}\!-\!\frac{J^{2}\big(t|a|^{-\frac{1}{2}}\big)}{t^{2}}\bigg](f(t;a)\!+\!1)\!+\!\frac{f^{3}(t;a)\!+\!3f^{2}(t;a)}{t^{2}},
\end{eqnarray}
with the initial condition $f(t;a)\sim-t^{2}$, $t\rightarrow0$. From Step 1, when $|a|$ large enough, we have
\begin{eqnarray}
\label{3.18}
-1<f(\rho;a)\leq 0,~~~~~\rho\in[0,\delta],
\end{eqnarray}
where $\delta$ is defined in $\eqref{3.14}$. Then, we can get
\begin{eqnarray}
\label{3.19}
-1<f(t;a)\leq 0,~~~~~t\in[0,\delta\sqrt{|a|}].
\end{eqnarray}
Furthermore, it is worth stressing that $\rho^{-1-k}J(r),\rho^{-k}h(r)\leq R^{*}$, for $0\leq\rho\leq1$, then we obtain the estimate
\begin{eqnarray}
\label{3.20}
\bigg|\frac{h^{2}\big(t|a|^{-\frac{1}{2}}\big)}{|a|}\!-\!\frac{J^{2}\big(t|a|^{-\frac{1}{2}}\big)}{t^{2}}\bigg|
\leq2\frac{t^{2k}{R^{*}}^{2}}{|a|^{1+k}},
~~~~~t\in[0,\sqrt{|a|}].
\end{eqnarray}
Moreover, for $t\in[0,\delta\sqrt{|a|}]$, we obtain
\begin{eqnarray}
\label{3.21}
\bigg|\frac{h^{2}\big(t|a|^{-\frac{1}{2}}\big)}{|a|}\!-\!\frac{J^{2}\big(t|a|^{-\frac{1}{2}}\big)}{t^{2}}\bigg|
\leq2\frac{t^{2k}{R^{*}}^{2}}{|a|^{1+k}}
\leq2\frac{\delta^{2k}{R^{*}}^{2}}{|a|}.
\end{eqnarray}
According to
\begin{equation*}
2\frac{\delta^{2k}{R^{*}}^{2}}{|a|}\rightarrow 0
~~as~~|a|\rightarrow \infty,
\end{equation*}
we can get that $\eqref{3.17}$ is equal to
\begin{eqnarray}
\label{3.22}
\frac{d^{2}f(t)}{dt^{2}}=\frac{2f(t)}{t^{2}}
+\frac{f^{3}(t)+3f^{2}(t)}{t^{2}},
\end{eqnarray}
with $f(t)\sim-t^{2}$ as $t\rightarrow0$. Now, we observe that $f(t)$ can be extended until $f(t)$ becomes infinite. Let $s=\ln t$, then we have
\begin{eqnarray}
\label{3.23}
f''(s)-f'(s)=2f(s)+f^{3}(s)+3f^{2}(s),
\end{eqnarray}
with $f(-\infty)=f'(-\infty)=0$. Multiplying $\eqref{3.23}$ by $f'(s)$ and integrating from $-\infty$ to $s$, we obtain
\begin{eqnarray}
\label{3.24}
(f'(s))^{2}=2f^{2}(s)\left(\frac{1}{2}f(s)+1\right)^{2}+2\int_{-\infty}^s
(f'(\alpha))^{2}\dd\alpha.
\end{eqnarray}
If there holds $f'(s_{0})=0$ for some $s_{0}>-\infty$, then applying $\eqref{3.24}$ we have $f'(s)\equiv0$ for $s\in(-\infty,s_{0})$, which bring about $f(s)\equiv0$. This contradicts the fact in Step 1. Besides,  since $f'(s)<0$ initially, then we get $f'(s)<0$, $\forall s\in(-\infty,\infty)$. Hence, for given $s_{0}\in(-\infty,\infty)$, we have
\begin{eqnarray*}
f'(s)\leq-\sqrt{2\int_{-\infty}^{s_{0}}
(f'(\alpha))^{2}\dd\alpha}<0,~~~s>s_{0}.
\end{eqnarray*}
Above inequality argues that $f(s)\rightarrow-\infty$ as $s\rightarrow\infty$ and $f(s)$ must cross $-1$ at some finite point $s_{0}=lnt_{0}$, the solution of $\eqref{3.16}$ decrease monotonically crosses $-1$ at $t$ close $t_{0}$. Thus we can conclude the fact that $\mathcal{A}_{2}$ is nonempty. Obviously, using the continuous dependence of the solution on the parameter $a$, we implies that $\mathcal{A}_{2}$ is open.

According to connectedness, since $\mathcal{A}_{1}$ and $\mathcal{A}_{2}$ are two open disjoint nonempty sets, then there exists at least one $a<0$ and $a\in\mathcal{A}_{3}$ $($say $a_{0}\,)$ such that $f(\rho;a)$ satisfies conditions $f'(\rho;a_0)\leq0,\,-1<f(\rho;a_0)<0,\,\forall \rho>0$. We claim now that $f(\rho;a_0)$ is a solution of $\eqref{3.1}$ satisfying $f(0;a_0)=0,\,f(\infty;a_0)=-1$.

In what follows, we show that the solution $f(\rho;a_0)$ satisfies $f(\infty;a_0)=-1$. Since $a_0\in\mathcal{A}_{3}$, then $f(\rho;a_0)$ satisfies $f'(\rho;a_0)\leq0,\,-1<f(\rho;a_0)<0,\,\forall \rho>0$. It is clearly that $f(\rho;a_0)$ is decreasing and bounded. Therefore, we have that $\lim_{\rho\rightarrow\infty}f(\rho;a_{0})\triangleq L\geq-1$. We are now ready to prove $L=-1$. Suppose by contradiction that $L>-1$. Hence, by $\eqref{3.1}$, we obtain $\lim_{\rho\rightarrow\infty}f''(\rho;a_{0})=L+1>0$. Then, there exists a certain $R>0$ large enough so that when $\rho>R$ there holds $f'(\rho;a_0)>0$. This contradicts the fact that $f'(\rho;a_0)\leq0$ for all $\rho>0$ and concludes the proof.

In summary, we find a solution $F(\rho)=f(\rho)+1$ satisfying $\eqref{2.2}$ subject to the conditions $F(0)=1,\,F(\infty)=0$.

Step 3{\rm:} Uniqueness of solution. Note that the solution $F(\rho;a)$ for given $a_0$ is unique. If otherwise, assume that there are two different solutions $F_1(\rho)$, $F_2(\rho)$ and set $\Psi(\rho)=F_2(\rho)-F_1(\rho)$ which satisfies the boundary conditions $\Psi(0)=\Psi(\infty)=0$ and the equation
\begin{eqnarray}
\label{3.25}
\rho^{2}\Psi''(\rho)&=&\left[\rho^{2}h^{2}(\rho)-J^{2}(\rho)+F_{1}^{2}(\rho)-1+F_{2}^{2}(\rho)+F_{1}(\rho)F_{2}(\rho)\right]\Psi(\rho)\notag\\[1mm]
&\triangleq& Q(\rho)\Psi(\rho),\,\,\,0<\rho<+\infty.
\end{eqnarray}
Without loss of generality, let $\Psi(\rho)>0$ as $\rho>0$ small enough. Together with $\eqref{2.2}$, we have
\begin{eqnarray}
\label{3.26}
  r^{2}F_{i}''(\rho)&=&\left[\rho^{2}h^{2}(\rho)-J^{2}(\rho)+F_{i}^{2}(\rho)-1\right]F_{i}(\rho)\notag\\[1mm]
&\triangleq& q_{i}(\rho)F_{i}(\rho),\,\,0<\rho<+\infty,\,\,(i=1,2).
\end{eqnarray}
Since $Q(\rho)-q_{i}(\rho)>0$, applying the Sturm comparison
theorem to $\eqref{3.25}$ and $\eqref{3.26}$, we get that $F_{i}(\rho)$ have
more zero points than $\Psi(\rho)$. We observe that $F_{i}(\rho)\neq0$ as
$\rho\in[0,+\infty)$. Hence, we have $\Psi(\rho)\neq0$ for all
$\rho\in[0,+\infty)$, which contradicts $\Psi(0)=0$ and concludes the proof.

Step 4{\rm:} The other properties of solution.

We argue that, if $a_{0}\in\mathcal{A}_{3}$, there is a constant $N>0$ such that $|a_{0}|^{1+k}\leq N{R^{*}}^{2}$. Then, for $\rho\leq1$, the solution $F(\rho;a_0)$ of $\eqref{2.2}$ satisfies
\begin{eqnarray}
\label{3.27}
|\rho^{-2}f(\rho)|\leq R^{**},
\end{eqnarray}
where $R^{**}>0$ is a constant depending on $R^{*}$ and $N$, $f(\rho)=F(\rho)-1$. Besides, we have $F'(0)=0$. In fact, we have that the solution $f(t;a)$ of $\eqref{3.17}$ crosses $-1$ if the term $\frac{h^{2}(t|a|^{-\frac{1}{2}})}{|a|}-\frac{J^{2}(t|a|^{-\frac{1}{2}})}{t^{2}}$ is enough small by Step 2. According to $\eqref{3.20}$ and the definition $t_{0}$ in Step 2, there is a $\varepsilon>0$ satisfying $2\frac{t_{0}^{2k}{R^{*}}^{2}}{|a|^{1+k}}<\varepsilon$ such that $f(t;a)$ crosses $-1$ near $t_{0}$. Letting $N=2t_0^{2k}/\varepsilon$, thus if $|a|^{1+k}>N{R^{*}}^{2}$, we obtain $a\in\mathcal{A}_{2}$. But $a_{0}\notin\mathcal{A}_{2}$, thus $|a_{0}|^{1+k}\leq N{R^{*}}^{2}$. Using $\eqref{3.5}$, $|a_{0}|^{1+k}\leq N{R^{*}}^{2}$, together with the assumption on $J(\rho)$ and $h(\rho)$, we arrive at
\begin{eqnarray}
\label{3.28}
\left|\frac{f(\rho;a_0)}{\rho^{2}}\right|
&\leq&|a_{0}|+\frac{2}{3}\int_0^\rho\frac{1}{r}\bigg|\bigg(h^{2}(r)-\frac{J^{2}(r)}{r^{2}}\bigg)(f(r;a_0)+1)+\frac{3f^{2}(r;a_0)+f^{3}(r;a_0)}{r^{2}}\bigg|\dd r\notag\\[1mm]
&\leq&|a_{0}|+\frac{2}{3}\int_0^{\rho}\frac{1}{r}\left|2r^{2k}{R^{*}}^{2}+3|a_0|^{2}r^{2}\right|\dd r\notag\\[2mm]
&\leq&(N{R^{*}}^2)^{\frac{1}{k+1}}+(N{R^{*}}^2)^{\frac{2}{k+1}}+\frac{2}{3k}{R^{*}}^{2}\triangleq R^{**},\,\,\,\,\,\forall \rho\leq1.
\end{eqnarray}
Hence $\eqref{3.27}$ follows. It is obvious that $\rho^{-2}f(\rho)$ increasing in $\rho$. Then applying the above result, we have $f(\rho)=O(\rho^{2})(\rho\rightarrow0)$. By the definition of the derivative of the function $f(\rho)$ at $\rho=0$, we get
\begin{equation*}
f'(0)=\lim\limits_{\rho\rightarrow0}\frac{f(\rho)-f(0)}{\rho}=\lim\limits_{\rho\rightarrow0}\frac{O(\rho^{2})}{\rho}=0.
\end{equation*}
Similarly, $F'(0)=0$.

We observe now that, there holds the asymptotic decay estimates $F(\rho)=O(\e^{-(1-\varepsilon)\rho})$ as $\rho\to\infty$, where $0<\varepsilon<1$ can be taken to be arbitrarily small. Define the comparison function $\eta(\rho)=Ce^{-(1-\varepsilon)\rho}$, where $0<\varepsilon<1$, $C>0$ is a constant to be chosen. Applying $\eqref{2.2}$ and $h(\rho)>0$, we have that $\forall\varepsilon>0$, there exists a $\rho_{\varepsilon}>0$ large enough such that
\begin{eqnarray}
\label{3.29}
\left(F-\eta\right)''
&=&h^{2}F+\frac{1}{\rho^{2}}F\left(F^{2}-1-J^{2}\right)-\left(1-\varepsilon\right)^{2}\eta\notag\\[1mm]
&=&\left(1-\varepsilon\right)^{2}\left(F-\eta\right)+\left[1-\left(1-\varepsilon\right)^{2}\right]F+\frac{1}{\rho^{2}}F\left(F^{2}-1-J^{2}\right)\notag\\[1mm]
&\geqslant&\left(1-\varepsilon\right)^{2}\left(F-\eta\right),\,\rho>\rho_{\varepsilon}.
\end{eqnarray}
Let $C>0$ be sufficiently large to make $(F-\eta)(\rho_{\varepsilon})\leq 0$. Then, according to the boundary condition $(F-\eta)(\rho)\rightarrow0\left(\rho\rightarrow\infty\right)$, we obtain $0<F\leq\eta=Ce^{-(1-\varepsilon)\rho},\,\rho>\rho_{\varepsilon}$ in view of the maximum principle and conclude the proof.

To sum up, the conclusion follows exactly as in Lemma \ref{le3.1}.
\end{proof}

\subsection{The equation governing the $h$--component}

In this subsection, we pay our attention to the equation that governs the $h$--component of a solution to the Julia--Zee dyon equations.

\begin{lemma}\label{le4.1}
Given $J(\rho)$ and $h(\rho)$ as in Lemma \ref{le3.1}, and the associated function $F(\rho)$, we can find a unique continuously differentiable solution $\tilde{h}(\rho)$ satisfying
\be\label{4.1}
\tilde{h}''\rho^2+2\tilde{h}'\rho=2F^2\tilde{h}+\frac{\beta^2}{2}\tilde{h}(\tilde{h}^2-1)\rho^2,~~~\rho>0,
\ee
subject to the condition
\be\label{4.2}
\tilde{h}(0)=0,~~~\tilde{h}(\infty)=1,
\ee
along with $\tilde{h}(\rho)$ increasing and $\rho^{-1}\tilde{h}(\rho)$ decreasing in $\rho$ and bounded as $\rho\to 0$. Besides, $\rho^{-1}\tilde{h}(\rho)\leq M$ for $\rho<1$, where $M=M(R^{**},\beta,N)>0$ and $N$ is a positive constant. Furthermore, we have the sharp asymptotic estimates
\bea
\tilde{h}(\rho)&=&b\rho+O(\rho^2),~~~0<b<N(2R^{**}+\frac{\beta}{2}),~~~\rho\to 0,\label{22}\\
\tilde{h}(\rho)&=&1+O(\e^{-\beta(1-\varepsilon)\rho}),~~~\rho\to\infty,~~~\varepsilon>0,~~~\beta>0.\label{33}
\eea
\end{lemma}

\begin{proof}
Step 1{\rm:} The existence and uniqueness of the local solution to the initial value problem. Let $H(\rho)=\rho\tilde{h}(\rho)$, then since $H(0)=0$, we can transform $\eqref{4.1}$ into the form
\begin{eqnarray}
\label{4.3}
H''=\frac{2F^{2}H}{\rho^{2}}+\frac{\beta^{2}}{2}H({\tilde{h}}^{2}-1)=\frac{2F^{2}H}{\rho^{2}}+\frac{\beta^{2}}{2}H(\frac{H^{2}}{\rho^{2}}-1),~~\rho>0.
\end{eqnarray}
Notice that the two linearly independent solutions of equation $H''-\frac{2H}{\rho^{2}}=0$ are $H=\rho^{2},\rho^{-1}$. Then we can convert $\eqref{4.3}$, for $\rho>0$ small enough, at least formally, into the integral equation
\begin{eqnarray}
\label{4.4}
H(\rho;b)\!=\!b\rho^{2}\!+\!\dfrac{1}{3}\int_0^\rho\left(\frac{\rho^{2}}{r}\!-\!\frac{r^{2}}{\rho}\right)\bigg[\frac{2}{r^{2}}\left(F^{2}(r)\!-\!1\right)\!+\!\frac{\beta^{2}}{2}\left(\frac{H^{2}(r;b)}{r^{2}}\!-\!1\right)\bigg]H(r;b)\dd r,
\end{eqnarray}
which can be easily solved by Picard iteration for an arbitrary constant $b$, at least near $\rho=0$. Analogous to the Step 1 in Lemma \ref{le3.1}, it follows that there exists a locally continuous solution $H(\rho)=b\rho^{2}+O(\rho^{3})\,(\rho\rightarrow0)$ of the initial value problem consisting of $\eqref{4.3}$ and $H(0)=0$. A standard theorem on continuation of solutions assures us that the solution can be continued as a function of $\rho$ until $H$ becomes infinite. Further, applying the continuous dependence of the solution on the parameters theorem we acquire that the solution $H(\rho;b)$ depends continuously on the parameter $b$.

Step 2{\rm:} The existence of the global solution to the boundary value problem. Applying the standard theorem on continuation of solutions, the solution $H(\rho;b)$ as a function of $\rho$ can be extended to its maximal existence interval $[0,R_{b})$ with $R_b$ being infinite or finite.

With the discussion of Step 1, we are interested in $b>0$ and now we define $\mathcal{B}=\{b|b>0\}$ and the shooting sets $\mathcal{B}_{1},\mathcal{B}_{2},\mathcal{B}_{3}$ as follows
\begin{eqnarray*}
\mathcal{B}_{1}&=&\left\{b>0:\,\exists \rho>0\,\,\mbox{such that}\,\,H'(\rho;b)<0\,\,\mbox{before}\,\,\tilde{h}(\rho;b)\,\,\mbox{becomes infinite}\right\},\\
\mathcal{B}_{2}&=&\left\{b>0:\,H(\rho;b)\,\,\mbox{becomes infinite before}\,\,H'(\rho;b)\,\,\mbox{becomes zero}\right\},\\
\mathcal{B}_{3}&=&\left\{b>0:\,\tilde{h}(\rho;b)\,\,\mbox{is bounded and}\,\,H'(\rho;b)\geq0\,\,\mbox{for all}\,\rho>0\right\}.
\end{eqnarray*}
It is obvious that $\mathcal{B}_{1}\cup\mathcal{B}_{2}\cup\mathcal{B}_{3}=\mathcal{B},\,\mathcal{B}_{1}\cap\mathcal{B}_{2}=\mathcal{B}_{2}\cap\mathcal{B}_{3}=\mathcal{B}_{3}\cap\mathcal{B}_{1}=\emptyset$.

For the purpose of proving that the set $\mathcal{B}_{3}$ is nonempty, we can first demonstrate that the sets $\mathcal{B}_{1},\mathcal{B}_{2}$ are both open and nonempty. We start by showing that $\mathcal{B}_{1}$ contains small $b$. In view of $\tilde{h}(\rho;0)\equiv0$ in $\eqref{4.4}$, then for any bounded interval $[\rho_1,\rho_2]\subset(0,\infty)$, there exists a $\delta>0$ so that
\be\label{4.5}
0<\tilde{h}(\rho;b)<\frac{1}{\sqrt{2}},~~b\in(0,\delta),~~\rho\in[\rho_1,\rho_2].
\ee
Then, by straightforward calculations we get
\begin{equation*}
-\frac{\beta^{2}}{2}<\frac{2F^{2}(\rho)}{\rho^{2}}+\frac{\beta^{2}}{2}(\tilde{h}(\rho;b)-1)\leq-\frac{\beta^{2}}{8},~~\rho\in[\rho_1,\rho_2],~~\rho\geq\frac{4}{\beta}.
\end{equation*}
From $\eqref{4.3}$, we have
\begin{eqnarray}
\label{4.6}
H''=\bigg[\frac{2F^{2}}{\rho^{2}}+\frac{\beta^{2}}{2}(\tilde{h}^{2}-1)\bigg]H\leq-\frac{\beta^{2}}{8}H,~~\rho\in[\rho_1,\rho_2],~~\rho\geq\frac{4}{\beta}.
\end{eqnarray}
Noting that equation $\eqref{4.6}$ is an oscillatory equation in the above bounded range of $\rho$ and distance between neighboring zeros does not exceed $\frac{2\sqrt{2}\pi}{\beta}$. Without prejudice to generality, we can assume $\rho_1$ and $\rho_2$ are adjacent zeros in $\eqref{4.5}$ satisfying $\rho_1,\,\rho_2\geq\frac{4}{\beta}$, then there is an $\rho_0\in(\rho_1,\rho_2)$ so that $H'(\rho_0;b)=0$. In fact, from $\eqref{4.6}$, whether for any $\rho\in(\rho_0,\rho_2)$ or $\rho\in(\rho_1,\rho_0)$, we can all obtain
\begin{eqnarray}
\label{4.7}
|H'(\rho;b)|=\bigg|\int_{\rho_{i}}^\rho\bigg(\frac{2F^{2}(r)}{r^{2}}+\frac{\beta^{2}}{2}(\tilde{h}^{2}(r;b)-1)\bigg)H(r;b)\dd r\bigg|\leq\frac{\pi^{2}}{2\sqrt{2}},~~i=0,1.
\end{eqnarray}
According to mean value theorem, together with $\eqref{4.7}$, whence we have
\begin{eqnarray*}
|H(\rho;b)|=|H'(\zeta;b)(\rho-\rho_1)|\leq\frac{\pi^{3}}{\beta},~~\rho\in(\rho_1,\rho_2),~~\rho_1<\zeta<\rho.
\end{eqnarray*}
In summary, for a sufficiently small $b\in(0,\delta)$, it can be concluded that there exists a $\rho>0$ such that $H'(\rho)<0$ before $\tilde{h}(\rho;b)$ becomes infinite. Consequently $\mathcal{B}_{1}$ is nonempty.

In order to proof $\mathcal{B}_{2}$ is nonempty, we use the variable $t$ to replace $\rho:\rho=b^{-\frac{1}{2}}t$. Thus \eqref{4.3} becomes
\begin{eqnarray}
\label{4.8}
H''(t)=\frac{2}{t^{2}}F^{2}(b^{-\frac{1}{2}}t)H(t)+\frac{\beta^{2}}{2t^{2}}H^{3}(t)-\frac{\beta^{2}}{2b}H(t).
\end{eqnarray}
Analogously to the second step in Lemma \ref{le3.1}, if $b\rightarrow\infty$ in above equation, together with \eqref{3.27}, then for $(2R^{**}+\frac{\beta}{2})b^{-1}\leq\frac{1}{N}$, we obtain
\begin{equation*}
\frac{2}{t^{2}}F^{2}(b^{-\frac{1}{2}}t)H(t)-\frac{\beta^{2}}{2b}H(t)\geq\frac{2}{t^{2}}H(t)-(2R^{**}+\frac{\beta^{2}}{2})b^{-1}H(t)>0
\end{equation*}
and \eqref{4.8} can be rewritten as $H''(t)>\frac{\beta^{2}}{2t^{2}}H^{3}(t)$. Hence, for all $t\in(0,1)$, we get $H(t)>0,\,H'(t)>0,\,H''(t)>0$. On the other hand, letting $t=\mathrm{e}^{s}$, then we have $H''(s)-H'(s)>\frac{\beta^{2}}{2}H^{3}(s)$. Multiplying above inequality by $H'(s)$ and integrating from $-\infty$ to $s$, then we get
\begin{eqnarray*}
\frac{1}{2}({H'(s)})^{2}>\frac{\beta^{2}}{8}H^{4}(s)+\int_{-\infty}^s ({H'(s)})^{2}\dd\tau>\frac{\beta^{2}}{8}H^{4}(s).
\end{eqnarray*}
Moreover, we have ${H'(s)}>\frac{\beta}{2}H^{2}(s),~~\forall s>-\infty$. Then, since
\begin{eqnarray*}
-\frac{1}{H(\tau)}\bigg|_{s_1}^s=\int_{s_1}^s\frac{H'(\tau)}{{H^{2}(\tau)}}\dd\tau>\frac{\beta}{2}(s-s_1),~~s>s_1>-\infty,
\end{eqnarray*}
accordingly, $H(s)$ blows up at some finite point $s_1=\ln t_1$. Indeed we have shown that $H(\rho;b)$ becomes infinite before $H'(\rho;b)$ becomes zero. Therefore $\mathcal{B}_{2}$ is nonempty. Continuity of the parameter $b$ can ensure that two sets are open.

According to connectedness, since $\mathcal{B}_{1}$ and $\mathcal{B}_{2}$ are two open disjoint nonempty sets, then there must be some value of $b$ in neither $\mathcal{B}_{1}$ nor $\mathcal{B}_{2}$ such that $\tilde{h}(\rho;b)$ satisfies conditions $\tilde{h}(\rho;b)$ is bounded and $H'(\rho;b)\geq0$ for all $\rho>0$. In other words, if $b\in\mathcal{B}_{2}$ $($say $b_{0}\,)$, we can get
\begin{equation}
\label{4.9}
b_{0}\leq N\left(2R^{**}+\frac{\beta}{2}\right).
\end{equation}
We observe now that $\tilde{h}(\rho;b_0)$ is a solution of $\eqref{4.1}$ satisfying $\tilde{h}(0;b_0)=0,\,\tilde{h}(\infty;b_0)=1$.

Since $b_0\in\mathcal{B}_{3}$, then the solution $\tilde{h}(\rho;b_{0})$ is bounded at $[0,+\infty)$. We next claim that the solution $\tilde{h}(\rho;b_0)$ satisfies $\tilde{h}(\infty;b_0)=1$. To this end, we first prove that $\tilde{h}(\rho;b_0)\leq1,\,\forall\rho>0$. If otherwise $\tilde{h}(\rho;b_0)>1,\,\exists\rho>0$. Then we argue that $\lim_{\rho\rightarrow\infty}\tilde{h}'(\rho;b_0)=0$. In fact, setting $\lim_{\rho\rightarrow\infty}\tilde{h}'(\rho;b_0)=\alpha>0$, then there exists a $\rho>0$ large enough such that $\tilde{h}'(\rho;b_0)>\frac{\alpha}{2}>0$. Furthermore, $\tilde{h}(\rho;b_0)>\frac{\alpha}{2}\rho+C$, which contradicts the finiteness of $\tilde{h}(\rho;b_0)$. Replacing the above result with the equation \eqref{4.1}, we observe that $\tilde{h}''(\rho;b_0)>0$ as $\rho$ sufficiently large, which contradicts $\lim_{\rho\rightarrow\infty}\tilde{h}'(\rho;b_0)=0$. We show now that for any $\varepsilon>0$, there is a $R_{\varepsilon}>0$ so that $\tilde{h}(\rho;b_0)>1-\varepsilon$ as $\rho>R_{\varepsilon}$. Suppose by contradiction that $\tilde{h}(\rho;b_0)\leq1-\varepsilon$ for all $\rho>\sqrt{\frac{8}{\beta^{2}\varepsilon}}$. Substituting the above hypothesis into the equation \eqref{4.3} yields
\begin{equation}
\label{4.10}
H''(\rho;b_0)\leq\bigg(\frac{2}{\rho^{2}}-\frac{\beta^{2}}{2}\varepsilon\bigg)H(\rho;b_0)\leq-\frac{\beta^{2}}{4}\varepsilon H(\rho;b_0),~~\rho>\sqrt{\frac{8}{\beta^{2}\varepsilon}}.
\end{equation}
Then $H(\rho;b_0)$ is an oscillating function as $\rho\in\big[\sqrt{\frac{8}{\beta^{2}\varepsilon}},\infty\big)$, which contradicts $b_0\in\mathcal{B}_{3}$. In other words, there exists a $\rho_0>\sqrt{\frac{8}{\beta^{2}\varepsilon}}$ such that $\tilde{h}(\rho_0;b_0)>1-\varepsilon$. Moreover, we stress that there is no minimum of the solution $\tilde{h}(\rho;b_0)$ less than $1-\varepsilon$ for all $\rho>\rho_0$. In reality, if not, there is a $\rho_1>\rho_0$ so that $\tilde{h}(\rho_1;b_0)\leq1-\varepsilon$ and $\tilde{h}'(\rho_1;b_0)=0$. Clearly, $\tilde{h}''(\rho;b_0)<0$ in \eqref{4.1}, then $\tilde{h}'(\rho;b_0)\leq0$ for all $\rho>\rho_1$. Furthermore, $\tilde{h}(\rho;b_0)\leq1-\varepsilon$ for all $\rho>\rho_1$. Similarly, this contradicts the fact that $b_0\in\mathcal{B}_{3}$. With the above analysis, we must have $1-\varepsilon<\tilde{h}(\rho;b_{0})\leq1$ as $\rho\rightarrow{+\infty}$. In conclusion, we obtain $\lim_{\rho\rightarrow\infty}\tilde{h}(\rho;b_{0})=1$.

In summary, we find a solution $\tilde{h}(\rho)$ satisfying $\eqref{4.1}$ subject to the conditions $\tilde{h}(0)=0,\,\tilde{h}(\infty)=1$.

Step 3{\rm:} Uniqueness of solution. We claim that the solution $\tilde{h}(\rho)$ for above given parameter $b_{0}$ is unique. Suppose otherwise that there are two different solutions $\tilde{h}_{1}(\rho),\tilde{h}_{2}(\rho)$. Then setting $\Psi(\rho)=\tilde{h}_{2}(\rho)-\tilde{h}_{1}(\rho)$ which satisfies the boundary condition $\Psi(0)=\Psi(\infty)=0$ and the equation
\begin{eqnarray}
\label{4.11}
(\rho\Psi(\rho))''&=&\left[\frac{\beta^{2}}{2}\left(\tilde{h}_{2}^{2}(\rho)+\tilde{h}_{1}^{2}(\rho)+\tilde{h}_{1}(\rho)\tilde{h}_{2}(\rho)-1\right)+\frac{2}{\rho^{2}}F^{2}(\rho)\right](\rho\Psi(\rho))\notag\\[1mm]
&\triangleq&P(\rho)(\rho\Psi(\rho)),\,\,\,0<\rho<+\infty.
\end{eqnarray}
Without loss of generality, let $\Psi(\rho)>0$ when $\rho>0$ is sufficiently small. In view of $\eqref{4.1}$, we get
\begin{eqnarray}
\label{4.12}
(\rho\tilde{h}_{i}(\rho))''&=&\left[\frac{\beta^{2}}{2}\left(\tilde{h}_{i}^{2}(\rho)-1\right)+\frac{2}{\rho^{2}}F^{2}(\rho)\right](\rho\tilde{h}_{i}(\rho))\notag\\[1mm]
&\triangleq&p_{i}(\rho)(\rho\tilde{h}_{i}(\rho)),\,\,0<\rho<+\infty,\,\,(i=1,2).
\end{eqnarray}
Since $P(\rho)-p_{i}(\rho)>0$, applying the Sturm comparison theorem to $\eqref{4.11}$ and $\eqref{4.12}$, we have that $\rho\tilde{h}_{i}(\rho)$ have more zero points than $\rho\Psi(\rho)$. Noting that $\rho\tilde{h}_{i}(\rho)\neq0$ for all $\rho\in(0,+\infty)$, then we have $\Psi(\rho)\neq0$ at a finite internal of $\rho$. Multiplying the equation $\eqref{4.11}$ by $\rho\tilde{h}_{1}(\rho)$, and the equation $\eqref{4.12}$ $($take $i=1)$ by $\rho\Psi(\rho)$, then subtracting, we have
\begin{equation*}
(\rho\Psi(\rho))''(\rho\tilde{h}_{1}(\rho))-(\rho\tilde{h}_{1}(\rho))''(\rho\Psi(\rho))=\frac{\beta^{2}}{2}(\tilde{h}_{2}^{2}(\rho)+\tilde{h}_{1}(\rho)\tilde{h}_{2}(\rho))(\rho\Psi(\rho))(\rho\tilde{h}_{1}(\rho))>0.
\end{equation*}
Furthermore $[(\rho\Psi(\rho))'(\rho\tilde{h}_{1}(\rho))-(\rho\tilde{h}_{1}(\rho))'(\rho\Psi(\rho))]'>0$. In other words, $(\rho\Psi(\rho))'(\rho\tilde{h}_{1}(\rho))-(\rho\tilde{h}_{1}(\rho))'(\rho\Psi(\rho))$ is monotonically increasing. Applying
\begin{eqnarray*}
\left[(\rho\Psi(\rho))'(\rho\tilde{h}_{1}(\rho))-(\rho\tilde{h}_{1}(\rho))'(\rho\Psi(\rho))\right]\bigg|_{\rho=0}=\left[\rho^{2}\left(\Psi'(\rho)\tilde{h}_{1}(\rho)-\Psi(\rho)\tilde{h}_{1}'(\rho)\right)\right]\bigg|_{\rho=0}=0,
\end{eqnarray*}
we get $(\rho\Psi(\rho))'(\rho\tilde{h}_{1}(\rho))-(\rho\tilde{h}_{1}(\rho))'(\rho\Psi(\rho))>0$.
According to
\begin{equation*}
\left[\frac{\Psi(\rho)}{\tilde{h}_{1}(\rho)}\right]'
=\frac{(\Psi'(\rho)\tilde{h}_{1}(\rho)-\tilde{h}_{1}'(\rho)\Psi(\rho))}{\tilde{h}_{1}^{2}(\rho)}>0
\end{equation*}
and $(\tilde{h}_{1}^{-1}(\rho))\Psi(\rho)>0$ at $\rho=\varepsilon>0$, we easily obtain $(\tilde{h}_{1}^{-1}(\rho))\Psi(\rho)>0$ as $\rho\rightarrow\infty$, which contradicts
\begin{eqnarray*}
\lim\limits_{\rho\rightarrow\infty}\left(\frac{\Psi(\rho)}{\tilde{h}_{1}(\rho)}\right)=\lim\limits_{\rho\rightarrow\infty}\left(\frac{\tilde{h}_{2}(\rho)-\tilde{h}_{1}(\rho)}{\tilde{h}_{1}(\rho)}\right)=\lim\limits_{\rho\rightarrow\infty}\left(\frac{\tilde{h}_{2}(\rho)}{\tilde{h}_{1}(\rho)}-1\right)=0.
\end{eqnarray*}
Thus $\Psi\equiv0$, that is, $\tilde{h}_{1}(\rho)=\tilde{h}_{2}(\rho)$.

Step 4{\rm:} The other properties of solution. We now derive some further properties about the solution $\tilde{h}(\rho;b)$.

First, we claim that $\tilde{h}'(\rho)\geq0$ for all $\rho>0$. If otherwise, since $\tilde{h}'(\rho)>0$ initially, then there exists a maximum point $\rho_1$ such that $\tilde{h}'(\rho_1)=0$ and $\tilde{h}''(\rho_1)\leq0$. According to $\tilde{h}'(\rho)>0$ for $\rho$ sufficiently large, thus there exists a minimum point $\rho_2=\inf\{\,\rho\,|\,\tilde{h}'(\rho)=0,\, \rho>\rho_1\}$ so that $\tilde{h}'(\rho_2)=0$ and $\tilde{h}''(\rho_2)\geq0$. This contradicts the fact that $\tilde{h}''(\rho_2)\geq0$ can not happen in view of $\frac{2F^{2}(\rho)}{\rho^{2}}$ and $\tilde{h}^{2}(\rho)$ are decreasing in $[\rho_1,\rho_2]$. Therefore $\tilde{h}(\rho)$ is increasing for all $\rho>0$. Next using $\eqref{4.4}$ and
$H(\rho)>0,\,0\leq h(\rho),F(\rho)\leq1$ for
all $\rho>0$, we can get that $\rho^{-2}H(\rho)$ is decreasing as $\rho>0$, that is $\rho^{-1}\tilde{h}(\rho)$ is decreasing as $\rho>0$

In the following, we prove that if $b_{0}\in\mathcal{B}_{3}$, for $\rho\leq1$, we can find a suitably large constant $M$ such that $|\rho^{-1}\tilde{h}(\rho)|\leq M$. In order to justify this assertion, for all
$\rho\leq1$, applying $\eqref{3.27}$, $\eqref{4.4}$ and $\eqref{4.9}$, we arrive at
\begin{eqnarray}
\label{4.13}
|\tilde{h}(\rho)|
&\leq&b_0\rho+\frac{2}{3}\int_0^\rho\frac{\rho}{r}
\left|\frac{2}{r^{2}}(F^{2}(r)-1)+\frac{\beta^{2}}{2}\left(\frac{H^{2}(r)}{r^{2}}-1\right)\right|H(r)\dd r\notag\\[2mm]
&\leq&b_0\rho+\frac{2}{3}\int_0^\rho\frac{\rho}{r}
\left(2R^{**}+2\beta^{2}b_{0}^{2}\right)b_0r^2\dd r
=b_0\rho+\frac{2}{3}b_0\rho^3
\left(R^{**}+\beta^{2}b_{0}^{2}\right)\notag\\[2mm]
&\leq&\left\{N\left(2R^{**}+\frac{\beta}{2}\right)\left[1+\frac{2}{3}\left(R^{**}+N^{2}\beta^{2}\left(2R^{**}+\frac{\beta}{2}\right)^{2}\right)\right]\right\}\rho\triangleq M\rho.
\end{eqnarray}

Finally, we show that there holds the asymptotic decay estimates $\tilde{h}(\rho)=O(\e^{-\beta(1-\varepsilon)\rho})$ as $\rho\to\infty$, where $0<\varepsilon<1$ can be taken to be arbitrarily small. Setting $\omega(\rho)=\rho(\tilde{h}-1)$, then $\eqref{4.1}$ can be written as $\omega''=\beta^{2}\omega$ when $\rho\rightarrow\infty$. To get the estimate for $\tilde{h}(\rho)$ in $\eqref{4.1}$, we define the comparison function $\eta(\rho)=C\e^{-\beta(1-\varepsilon)\rho}$, where $C>0$ is a constant to be chosen later, $\varepsilon>0$ is sufficiently small. Then for any $\varepsilon>0$, there exists a sufficiently large $\rho_{\varepsilon}>0$ such that
\begin{eqnarray}
\label{4.14}
\left(\omega-\eta\right)''
&=&\frac{\beta^{2}}{2}\left(\tilde{h}+1\right)\tilde{h}\omega+\frac{2}{\rho}F^{2}\tilde{h}-\beta^{2}\left(1-\varepsilon\right)^{2}\eta\notag\\[1mm]
&\geq&\frac{\beta^{2}}{2}\left(\tilde{h}+1\right)\tilde{h}\left(\omega-\eta\right)+\left[\beta^{2}\left(1-\frac{\varepsilon}{2}\right)^{2}-\beta^{2}\left(1-\varepsilon\right)^{2}\right]\eta\notag\\[1mm]
&\geq&\frac{\beta^{2}}{2}\left(\tilde{h}+1\right)\tilde{h}\left(\omega-\eta\right),
\,\rho>\rho_{\varepsilon}.
\end{eqnarray}
We can choose $C>0$ be large enough to make $(\omega-\eta)(\rho_{\varepsilon})\leq0$. To get the other half of the estimate, we consider $\omega+\eta$ instead. In place of $\eqref{4.14}$, then we have the inequality
\begin{equation*}
\left(\omega+\eta\right)''\leq\frac{\beta^{2}}{2}\left(\tilde{h}+1\right)\tilde{h}\left(\omega+\eta\right)
\end{equation*}
for $\rho$ greater than some large $\rho_{\varepsilon}$. Hence, there holds $(\omega+\eta)(\rho_{\varepsilon})\geq0$, $\rho>\rho_{\varepsilon}$ when the coefficient $C$ in the definition of $\eta$ is made large enough. So the decay estimate for $\tilde{h}-1$ near infinity stated in Lemma \ref{le4.1} is established.

\end{proof}

\subsection{The equation governing the $J$--component}

This subsection is devoted to establish the existence and uniqueness of the solution of the $J$--component equation.

\begin{lemma}\label{le5.1}
Given $J(\rho)$ and $h(\rho)$ as in Lemma \ref{le3.1}, and the resulting function $F(\rho)$, we can find a unique continuously differentiable solution $\tilde{J}(\rho)$ satisfying
\be\label{5.1}
\tilde{J}''\rho^2=2\tilde{J}F^2,~~~\rho>0,
\ee
subject to the condition
\be\label{5.2}
\tilde{J}(0)=0,~~~\tilde{J}'(\infty)=C,
\ee
along with $\tilde{J}'(0)=0$, $\tilde{J}(\rho)$ increasing, $\tilde{J}'(\rho)\leq 1$, and $\rho^{-2}\tilde{J}(\rho)$ decreasing in $\rho$ and bounded as $\rho\to 0$. Besides, $\rho^{-2}\tilde{J}(\rho)\leq \sqrt{\frac{8R^{**}}{3}}C$ for $\rho<1$. Furthermore, we have the sharp asymptotic estimates
\bea
\tilde{J}(\rho)&=&c\rho^2+O(\rho^4),~~~0<c\leq\sqrt{\frac 83 R^{**}}C,~~~\rho\to 0,\label{44}\\
\tilde{J}''(\rho)&=&O(\e^{-2(\sqrt{1-C^2}-\varepsilon)\rho}),~~~\rho\to\infty,~~~\varepsilon>0.\label{55}
\eea
\end{lemma}

\begin{proof}
We can rewrite equation \eqref{2.4} in an integral form
\be\label{5.3}
\tilde{J}(\rho)=c\rho^2+\frac 23\int_0^\rho(r^{-1}\rho^2-r^2\rho^{-1})\frac{\tilde{J}(\rho)(F^2(\rho)-1)}{r^2}\dd r,
\ee
where $a$ is an arbitrary constant. According to a Picard iteration, we can get a local continuous solution of \eqref{5.1} with $\tilde{J}(0)=0$ for $\rho$ sufficiently small, satisfying
\be\label{5.4}
\tilde{J}(\rho)=c\rho^2+O(\rho^4),~~~\rho\to 0,
\ee
which continuously depends on the parameter $a$. Without loss of generality, we may assume $c>0$. The continuously differentiable of the solution implies that $\tilde{J}'(0)=0$. Otherwise $\tilde{J}'(0)\neq0$, then there exist a number $\delta>0$ and a constant $c_0>0$ such that
\be\label{5.5}
|\tilde{J}'(\rho)|>c_0,~~~\text{as}~~0<\rho<\delta.
\ee
In view of \eqref{5.1}, we see
\be\label{5.6}
\tilde{J}''\rho=\frac{2\tilde{J}F^2}{\rho}.
\ee
Taking limit on the both side of \eqref{5.6} and using \eqref{5.5}, we have
\be\label{5.7}
\lim_{\rho\to\infty}|\tilde{J}''\rho|=\lim_{\rho\to\infty}\left|\frac{2\tilde{J}F^2}{\rho}\right|=|2\tilde{J}'(0)|>2c_0,
\ee
namely,
\be\label{5.8}
|\tilde{J}''(\rho)|>\frac{2c_0}{\rho},~~~\rho\in(0,\delta),
\ee
which contracts with \eqref{5.4}. Consequently, $\tilde{J}'(0)=0$. Further, the equation \eqref{5.3} indicates that, if $c>0$, then $\tilde{J}(\rho)>0$ and $\tilde{J}'(\rho)>0$.

Assuming
\be\label{5.9}
\lim_{\rho\to\infty}\tilde{J}'(\rho)=C.
\ee
If $C$ is finite, then the linearity of the equation \eqref{5.1} for $\tilde{J}$ tells us that there is a unique $c$ satisfies \eqref{5.9}, and then
\be\label{5.10}
\tilde{J}'(\rho)<C,~~~\rho\in(0,\infty)
\ee
because of $\tilde{J}''(\rho)>0$.

In the following, we will show $C<\infty$. Note that $r<\rho$ and $F<1$, we see, the integral on the right-hand side of \eqref{5.3} is negative. Thus
\be\label{5.11}
\tilde{J}(\rho)<c\rho^2.
\ee
While for large $\rho$, the equation \eqref{2.2} indicates
\be\label{5.12}
F''\sim(1-C^2)F.
\ee
Multiplying the both side of \eqref{5.12} by $F'$ and integrating over $(\rho,\infty)$, we get $F'/F\sim-\sqrt{1-C^2}$, from which we obtain
\be\label{5.13}
F(\rho)=O(\e^{-(\sqrt{1-C^2}-\varepsilon)\rho}),~~~\rho\to\infty,~~\varepsilon>0.
\ee
Therefore, \eqref{5.1}, \eqref{5.11} and \eqref{5.13} indicate that
\be\label{5.14}
\tilde{J}''(\rho)=O(\e^{-2(\sqrt{1-C^2}-\varepsilon)\rho}),~~~\rho\to\infty,~~\varepsilon>0.
\ee
Thus, $\tilde{J}'(\infty)=C$ is finite. Particularly, $0\leq C<1$.

Differentiating \eqref{5.3}, we have
\be\label{5.15}
\tilde{J}'(\rho)=2c\rho+\frac 23\int_0^\rho(2\rho r^{-1}+\rho^{-2}r^2)\frac{\tilde{J}(r)(F^2(r)-1)}{r^2}\dd r,
\ee
then
\be\label{5.16}
|\tilde{J}'(\rho)-2c\rho|=\frac 23\left|\int_0^\rho(2\rho r^{-1}+\rho^{-2}r^2)\frac{\tilde{J}(r)(F^2(r)-1)}{r^2}\dd r\right|<\frac 83 c R^{**}\rho^3.
\ee
For large $c$, set $\rho_0=C/c$. If $c^2>\frac 83 R^{**}C^2$, then we know from \eqref{5.16} that, at $\rho_0$
\be\label{5.17}
|\tilde{J}'-2C|<\frac 83 c^{-2} R^{**}C^3<C,
\ee
which means $\tilde{J}'>C$. But in view of \eqref{5.10}, the critical value of $c$, say $c_0$, certainly satisfies
\be\label{5.18}
c_0^2\leq\frac 83 R^{**}C^2.
\ee
To proceed, \eqref{5.3} gives us that
\be\label{5.19}
\left(\frac{\tilde{J}(\rho)}{\rho^2}\right)'=\frac 23\int_0^\rho3r^2\rho^{-4}\frac{\tilde{J}(r)(F^2(r)-1)}{r^2}\dd r<0,~~~\rho\to 0,
\ee
so that
\be\label{5.20}
\tilde{J}(\rho)\leq c_0\rho^2,~~~\rho\to 0.
\ee
Hence, for $\rho\leq1$, we obtain
\be\label{5.21}
\tilde{J}(\rho)\leq \sqrt{\frac{8R^{**}}{3}}C\rho^2.
\ee
\end{proof}

\subsection{Proof of solution to the full equations}

With the lemmas established above, we are ready to start solving the complete Julia--Zee dyon equations \eqref{2.2}--\eqref{2.4} under the boundary conditions \eqref{2.5}--\eqref{2.6}. That is, we are now going to prove our main result, Theorem \ref{th2.1}.

We first define a Banach space $\mathscr{B}$ as follows
\be\nn
\mathscr{B}=\left\{\left(J(\rho),h(\rho)\right)|J(\rho),h(\rho)\in C\left([0,+\infty)\right),\rho^{-1-k}\left(1+\rho^{k}\right)J(\rho),\rho^{-k}\left(1+\rho^{k}\right)h(\rho)\,\mbox{are bounded}\right\}
\ee
with the norm
\begin{eqnarray*}
\left\|\left(J(\rho),h(\rho)\right)\right\|_{\mathscr{B}}=\sup_{\rho\in[0,+\infty)}\left\{|\rho^{-1-k}\left(1+\rho^{k}\right)J(\rho)|+|\rho^{-k}\left(1+\rho^{k}\right)h(\rho)|\right\},
\end{eqnarray*}
where $0<k<1$. There is no difficulty in checking that this is indeed a Banach space.

From Lemmas \ref{le3.1}--\ref{le5.1}, we observe that for given $\left(J(\rho),h(\rho)\right)\in C\left([0,+\infty)\right)$ satisfying $0<J(\rho),\,\rho h(\rho)\leq \rho^{1+k}R^{*}$ $(0<k<1)$ for all $0<\rho\leq1$ in Lemma \ref{le3.1}, we respectively get the unique corresponding increasing solutions $\tilde{J}(\rho),\,\tilde{h}(\rho)$ in Lemma \ref{le4.1} and \ref{le5.1} satisfying $\rho^{-1}\tilde{h}(\rho)\leq R^{***}$ and $\rho^{-2}\tilde{J}(\rho)\leq R^{***}$ for $\rho\leq1$, where $R^{***}=\max\{M,\sqrt{\frac{8R^{**}}{3}}C\}>0$ only depend on the choice of $\beta>0$, $R^{**}>0$, $0\leq C<1$ and a positive constant $N$. Noting that if $|J(\rho)|\leq \rho^2R^{***}$, $|h(\rho)|\leq \rho R^{***}$ for all $\rho\leq1$, then there are $\delta_{1}>0$, $\delta_{2}>0$ such that $|J(\rho)|\leq \rho^2R^{***}\leq\rho^{1+k}R^{*}$, $|h(\rho)|\leq \rho R^{***}\leq\rho^{1+k}R^{*}$ for $0<\rho<\delta=\min\{\delta_{1},\delta_{2}\}$. Next by the proceeding of proofs in \eqref{3.28}, \eqref{4.13} and \eqref{5.21}, we have $|\tilde{J}(\rho)|\leq \rho^2R^{***}$, $|\tilde{h}(\rho)|\leq \rho R^{***}$ for all $\rho\leq1$.

After that, we define the non-empty, bounded, closed, convex subset $\mathcal{S}$ of $\mathscr{B}$ as
\begin{eqnarray*}
&&\mathcal{S}=\big\{\left(J(\rho),h(\rho)\right)\in\mathscr{B}\,\big|~\,|\rho^{-2}J(\rho)|\leq R^{***},\,|\rho^{-1}h(\rho)|\leq R^{***},\,\forall\rho\leq1;\,J(\rho),\,h(\rho)\,\,\mbox{is increasing};\\[1mm]
&&~~~~~~~~~~~~~~~~~~~~~~~~~~~~~~~~~~~\,0<\frac{J(\rho)}{\rho}\leq C,\,0<h(\rho)\leq1,\,\forall\rho>0;\,h(\infty)=1,\,J'(\infty)=C\big\},
\end{eqnarray*}
where $R^{***}=R^{***}\{R^{**},N,\beta, C\}$ and $N$ is a positive constant. It is straightforward to examine that the set $\mathcal{S}$ is indeed nonempty, bounded, closed and convex.

It is worth stressing that the process of the Lemma \ref{le4.1} and Lemma \ref{le5.1} defines the mapping $T{\rm:}\,(J,h)\rightarrow(\tilde{J},\tilde{h})$ on $\mathcal{S}$ which maps $\mathcal{S}$ into itself. In the following, we will show that the mapping $T$ is continuous and compact. Then the Schauder fixed point theorem guarantees that $T$ has a fixed point, and the theorem \ref{th2.1} is proved.

In order to obtain that $T$ is continuous, we shall prove that if $\left\|\left((J_{1},h_{1})-(J_{2},h_{2})\right)\right\|_{\mathscr{B}}\rightarrow0$, then $\left\|T\left(J_{1},h_{1}\right)-T\left(J_{2},h_{2}\right)\right\|_{\mathscr{B}}\rightarrow0,\,\forall\left(J_{1},h_{1}\right),\left(J_{2},h_{2}\right)\in\mathcal{S}$. We can assume $T(J_{i})=\tilde{J}_{i}$, $T(h_{i})=\tilde{h}_{i}$, $i=1,2$. Obviously, $\tilde{h}_{1}\left(\infty\right)=1$, $\tilde{h}_{2}\left(\infty\right)=1$, then for any $\varepsilon>0$, there exists an $R>0$ adequately large so that
\begin{eqnarray}
\label{6.1}
\sup_{\rho\in[R,+\infty)}\left|\rho^{-k}\left(1+\rho^{k}\right)\left[\tilde{h}_{1}\left(\infty\right)-\tilde{h}_{2}\left(\infty\right)\right]\right|<\varepsilon.
\end{eqnarray}
Similarly, we can obtain
\begin{eqnarray}
\label{6.2}
&&\sup_{\rho\in[R,+\infty)}\left|\rho^{-1-k}\left(1+\rho^{k}\right)\left[\tilde{J}_{1}\left(\infty\right)-\tilde{J}_{2}\left(\infty\right)\right]\right|\notag\\[2mm]
&&<\sup_{\rho\in[R,+\infty)}\left|\rho^{-k}\left(1+\rho^{k}\right)\left[\tilde{J}_{1}\left(\infty\right)-\tilde{J}_{2}\left(\infty\right)\right]\right|<\varepsilon.
\end{eqnarray}
For above given $\varepsilon$, there exists a sufficiently small $\delta>0$ such that
\begin{eqnarray}
\label{6.3}
\sup_{\rho\in(0,\delta]}\left|\rho^{-k}\left(1+\rho^{k}\right)\left[\tilde{h}_{1}(0)-\tilde{h}_{2}(0)\right]\right|<\varepsilon,
\end{eqnarray}
in view of $\tilde{h}_{1}(0)=\tilde{h}_{2}(0)=0$. Likewise, we have
\begin{eqnarray}
\label{6.4}
\sup_{\rho\in(0,\delta]}\left|\rho^{-1-k}(1+\rho^{k})\left[\tilde{J}_{1}(0))-\tilde{J}_{2}(0)\right]\right|<\varepsilon.
\end{eqnarray}
According to the continuity of $T$ over $[\delta,R]$, then $T$ is continuous on $\mathscr{B}$.

In the following, we concentrate on the compact of $T$. To prove that the mapping $T$ is compact, we will demonstrate that the mapping $T$ map arbitrary bounded sets in $\mathcal{S}$ to precompact sets. That is, if $\left\{J_{n}(\rho),h_{n}(\rho)\right\}$ is an arbitrary bounded sequence in $\mathcal{S}$, then we show that $\{\tilde{J}_n(\rho),\tilde{h}_n(\rho)\}$ have a convergent sub--sequence in $\mathcal{S}$. Take $\{\tilde{h}_n(\rho)\}$ as an example. Applying Lemma \ref{le4.1} and $\{\tilde{h}_n(\rho)\}$ in $\mathcal{S}$, then by straightforward calculations we have
\begin{eqnarray*}
\left|(\rho\tilde{h}_n(\rho))'\right|&\leq&2b_{0}\rho+\dfrac{4}{3}\int_0^\rho\left|\frac{r^{2}}{\rho^{2}}\bigg[\frac{2}{r^{2}}\left(F^{2}(r)-1\right)+\frac{\beta^{2}}{2}\left(\frac{H^{2}(r)}{r^{2}}-1\right)\bigg]H(r)\right|\dd r\notag\\[2mm]
&\leq&2b_{0}\rho+\dfrac{4}{3}\int_0^\rho\frac{r^{2}}{\rho^{2}}\bigg[\frac{2}{r^{2}}\left(2r^{2}R^{**}\right)+\frac{\beta^{2}}{2}\bigg]r^{2}R^{***}\dd r\notag\\[2mm]
&\leq&2b_{0}\rho+\left(4R^{**}+\frac{\beta^{2}}{2}\right)\rho^{3}R^{***},
\end{eqnarray*}
that is, $(\rho\tilde{h}_n(\rho))'$ is bounded on any closed subinterval of $(0,\infty)$. In other words, there is an $M>0$ such that $|(\rho\tilde{h}_n(\rho))'|\leq M$ on any closed subinterval of $(0,\infty)$, where $M$ does not dependent of $n$. Using the mean value theorem, for any $\varepsilon>0,\,\rho_{1},\rho_{2}\in[\delta,R]\subset(0, \infty)$ there exists a $\delta=\frac{\varepsilon}{M+1}>0$ such that $|\rho_{1}\tilde{h}_n(\rho_{1})-\rho_{2}\tilde{h}_n(\rho_{2})|=|(\xi\tilde{h}_n(\xi))'||\rho_{1}-\rho_{2}|<\varepsilon$ as $|\rho_{1}-\rho_{2}|<\delta$, where $\xi$ is between $\rho_{1}$ and $\rho_{2}$. Thus $\{\rho\tilde{h}_n(\rho)\}$ is equicontinuity, so is $\{\tilde{h}_n(\rho)\}$. It is easy to see that $\{\tilde{h}_n(\rho)\}$ is uniformly bounded in view of  $\{\tilde{h}_n(\rho)\}$ in $\mathcal{S}$. By the Arzela--Ascoli theorem, there is a subsequence of $\{\tilde{h}_n(\rho)\}$ in $\mathcal{S}$ and the continuous function $\tilde{h}(\rho)$ such that the subsequence of $\{\tilde{h}_n(\rho)\}$ uniformly converges to $\tilde{h}(\rho)$ in any compact subinterval of $(0,\infty)$ $($denoted as $[\delta,R])$, hence we have
\begin{eqnarray}
\label{6.5}
\sup_{\rho\in[\delta,R]}\big|\rho^{-k}(1+\rho^{k})(\tilde{h}_n(\rho)-\tilde{h}(\rho))\big|\leq\left(1+\delta^{-k}\right)
\sup_{\rho\in[\delta,R]}\big|\tilde{h}_n(\rho)-\tilde{h}(\rho)\big|\rightarrow0\,(n\rightarrow\infty).
\end{eqnarray}
Similarly, we note that, given any sequence $\{J_{n}(\rho)\}$ in $\mathcal{S}$ that is bounded in the norm for $\mathscr{B}$. Then we have easy access to that $\{\tilde{J_{n}}'(\rho)\}$ is bounded in the usual norm for $\mathscr{B}$, so that $\{\tilde{J_{n}}(\rho)\}$ is bounded in any finite $\rho$-interval. According to the Arzela-Ascoli theorem, we can assure that there is a subsequence of $\{\tilde{J_{n}}(\rho)\}$ which converges uniformly in any such subinterval. Then the only question is whether for the limit functions, say $(\tilde{J}(\rho),\tilde{h}(\rho))$, it is true that
\begin{equation*}
||(\tilde{J_{n}}(\rho),\tilde{h}_n(\rho))-(\tilde{J}(\rho),\tilde{h}(\rho))||_{\mathscr{B}}\rightarrow0~~~\mbox{as}~~n\rightarrow\infty,
\end{equation*}
and the only possible difficulties occur as $\rho\rightarrow0$ or as $\rho\rightarrow\infty$ $($that is in interval $(0,\delta)$ and $(R,+\infty)$$)$.

As $\rho\rightarrow0$, since $\{\tilde{h}_n(\rho)\},\tilde{h}(\rho)\in\mathcal{S}$, then we get $\tilde{h}_n(\rho)\leq\rho R^{***},$\,$\tilde{h}(\rho)\leq\rho R^{***}$ for all $\rho\leq1$. Moreover, for given $\varepsilon>0$, we can find a $\delta=\left(\frac{\varepsilon}{4R^{***}}\right)^{\frac{1}{1-k}}$ small enough so that
\begin{eqnarray*}
\sup_{\rho\in(0,\delta)}|\rho^{-k}(1+\rho^{k})(\tilde{h}_n(\rho)-\tilde{h}(\rho))|<2\sup_{\rho\in(0,\delta)}|\rho^{1-k}\rho^{-1}(\tilde{h}_n(\rho)-\tilde{h}(\rho))|<4\delta^{1-k}R^{***}\leq\varepsilon.
\end{eqnarray*}
With $\delta>0$ fixed, we can then find $n$ large enough that
\begin{eqnarray}
\label{6.6}
\sup_{\rho\in(0,\delta)}\big|\rho^{-k}(1+\rho^{k})(\tilde{h}_n(\rho)-\tilde{h}(\rho))\big|\leq\varepsilon.
\end{eqnarray}
Likewise, for given $\varepsilon>0$, we can find a $\delta=\left(\frac{\varepsilon}{4R^{***}}\right)^{\frac{1}{1-k}}$ sufficiently small that
\begin{eqnarray*}
\sup_{\rho\in(0,\delta)}|\rho^{-1-k}(1+\rho^{k})(\tilde{J_{n}}(\rho)-\tilde{J}(\rho))|<2\sup_{\rho\in(0,\delta)}|\rho^{1-k}\rho^{-2}(\tilde{J_{n}}(\rho)-\tilde{J}(\rho))|<4\delta^{1-k}R^{***}\leq\varepsilon.
\end{eqnarray*}
With $\delta>0$ fixed, we can then find $n$ large enough that
\begin{eqnarray}
\label{6.7}
\sup_{\rho\in(0,\delta)}\big|\rho^{-1-k}(1+\rho^{k})(\tilde{J_{n}}(\rho)-\tilde{J}(\rho))\big|\leq\varepsilon.
\end{eqnarray}

As $\rho\rightarrow+\infty$, integrating $\eqref{4.3}$ from $\rho$ to $\infty$ and applying $\lim_{\rho\rightarrow\infty}(\rho\tilde{h}_n(\rho))'=1$, we get
\begin{equation}
\label{6.8}
(\rho\tilde{h}_n(\rho))'-1=-\int_\rho^{+\infty}\left[\dfrac{2}{r}F_{n}^{2}(r)\tilde{h}_n(r)+\dfrac{\beta^{2}}{2}r\tilde{h}_n(r)
(\tilde{h}_n^{2}(r)-1)\right]\dd r.
\end{equation}
Together with $\eqref{11}$ and $\eqref{33}$, we have
\begin{equation}
\label{6.9}
|\tilde{h}_n(\rho)-1|\leq|\rho\tilde{h}_n'(\rho)|+\left|\int_\rho^{+\infty}\left[\dfrac{2}{r}F_{n}^{2}(r)\tilde{h}_n(r)+\dfrac{\beta^{2}}{2}r
\tilde{h}_n(r)(\tilde{h}_n^{2}(r)-1)\right]\dd r\right|\rightarrow0,~\rho\rightarrow\infty.
\end{equation}
It is easy to see that $\sup_{\rho\in(R,+\infty)}|\rho^{-k}(1+\rho^{k})(\tilde{h}_n(\rho)-1)|\rightarrow0$ uniformly as $n\rightarrow\infty$ when we choose $R>0$ sufficiently large. As for $\{\tilde{J_{n}}(\rho)\}$, we have to prove that $\{\tilde{J_{n}}'(\rho)\}$ are uniformly convergent for $\rho\gg1$.
As $\rho\rightarrow+\infty$, integrating $\eqref{5.1}$ from $\rho$ to $\infty$ and applying $\lim_{\rho\rightarrow\infty}(\tilde{J_{n}}'(\rho))=C$, we get
\begin{equation}
\label{6.10}
\tilde{J_{n}}'(\rho)-C=-\int_\rho^{+\infty}\dfrac{2\tilde{J_{n}}(r)F_{n}^{2}(r)}{r^{2}}\dd r.
\end{equation}
Noting that
\begin{equation*}
\lim\limits_{\rho\rightarrow\infty}\frac{\tilde{J_{n}}(\rho)}{\rho}=\lim\limits_{\rho\rightarrow\infty}\tilde{J_{n}}'(\rho)=C,~~0<C<1.
\end{equation*}
Together with $\eqref{11}$, we have
\begin{equation}
\label{6.11}
|\tilde{J_{n}}'(\rho)-C|\leq\int_\rho^{+\infty}\left|\dfrac{2\tilde{J_{n}}(r)F_{n}^{2}(r)}{r^{2}}\right|\dd r\rightarrow0,~\rho\rightarrow\infty.
\end{equation}
It is easy to see that $\sup_{\rho\in(R,+\infty)}|\rho^{-1-k}(1+\rho^{k})(\tilde{J_{n}}'(\rho)-C)|\rightarrow0$ uniformly as $n\rightarrow\infty$ when we choose $R>0$ sufficiently large. In summary, the mapping $T$ is compact.

In addition, we may use \eqref{2.13} to calculate the electric charge of the solution in the original variables,
\bea
q_e&=&\frac{1}{4\pi}\lim_{r\to\infty}\oint_{|x|=r}\mathbf{E}\cdot \dd \mathbf{S}\nn\\
&=&\frac{1}{4\pi}\lim_{r\to\infty}\int_{|x|<r}\nabla\cdot\mathbf{E} \dd x\nn\\
&=&\frac{1}{4\pi g}\int_{\mathbb{R}^3}\partial_i\left\{\frac{x^i h(r)}{r}\frac{\dd}{\dd r}\left(\frac{J(r)}{r}\right)\right\}\dd x\nn\\
&=&\frac{1}{g}\int_0^\infty \left(r^2 h(r)\left(\frac{J(r)}{r}\right)'\right)'\dd r\nn\\
&=&\frac{1}{g}\int_0^\infty \left(rh'J'+rh J''-h'J\right)\dd r\nn\\
&=&\frac{1}{g}\int_0^\infty \left(\frac{2F^2Jh}{r}-h'J+rh'J'\right)\dd r.
\eea
Clearly, according to the asymptotic estimates of $F, J$ and $h$, we see the integral above is convergent.

The conclusion follows exactly as in Theorem \ref{th2.1}.
\hfill$\Box$\vskip7pt

\noindent{\textbf{DATA AVAILABILITY}}

The date that support the findings of this study are available from the corresponding author upon reasonable request.

\end{document}